\documentclass[11pt,a4paper,normalheadings,headsepline,headinclude,bibtotoc,fleqn,review]{article}
\usepackage[latin1]{inputenc}
\usepackage{amsmath}
\usepackage{booktabs}
\usepackage{array}
\usepackage{amsthm}
\usepackage{dcolumn}
\usepackage{endnotes}
\usepackage{units}
\usepackage{marvosym}
\usepackage[lines=10]{geometry}
\usepackage{longtable}
\usepackage{rotating}
\usepackage{expdlist}
\usepackage{geometry}
\usepackage{floatrow}
\usepackage{pgfplots}
\pgfmathdeclarefunction{gauss}{3}{
	\pgfmathparse{1/(#3*sqrt(2*pi))*exp(-((#1-#2)^2)/(2*#3^2))}%
}

\pgfmathsetmacro\valueA{gauss(2.25, 2.25, 0.7)} 

\newenvironment{distributionsaxes}{
	\begin{axis}[
		domain=0:6, 
		samples=100, 
		ymin=0, 
		ymax=0.75, 
		axis lines=left, 
		xlabel=$\ell_S{,}\ell_L$, 
		ylabel=$\Phi_S{,}\Phi_L$, 
		every axis y label/.style={at=(current axis.above origin),anchor=north east}, 
		every axis x label/.style={at=(current axis.right of origin),anchor=north west}, 
		height=7cm, 
		width=8cm, 
		xtick=\empty, 
		ytick=\empty, 
		enlargelimits=false, 
		clip=false, 
	]
}{
	\end{axis}
}

\newenvironment{linesaxes}{
	\begin{axis} [
		height=7cm, 
		width=8cm, 
		ymin=0, 
		ymax=100, 
		xmin = 0, 
		xmax = 100, 
	]
}{
	\end{axis}
}

\usepackage{wrapfig}
\usepackage{expdlist}
\usepackage{amssymb}
\usepackage{gloss}
\usepackage{nomencl}
\usepackage{natbib}

\newtheorem{lemma}{Lemma}

    \makegloss

\renewcommand{\thesection}{\arabic{section}}

\usepackage{fancyhdr} 
\begin{document}


\begin{center}\huge{On Market Design and Latency Arbitrage}\footnote{I thank Eric Budish and Al Roth for bringing the topics of market design and high-frequency trading to my attention. First draft August 2019.} \end{center}

\begin{center} \textit{Wolfgang Kuhle}\\ \textit{Zhejiang University, Hangzhou, China, e-mail wkuhle@gmx.de\\ MEA, Max Planck Institute for Social Law and Social Policy, Munich, Germany\\ VSE, Prague, Czech Republic}\end{center}

\noindent\emph{\textbf{Abstract:} We argue that contemporary stock market designs are, due to traders' inability to fully express their preferences over the execution times of their orders, prone to latency arbitrage. In turn, we propose a new order type which allows traders to specify the time at which their orders are executed after reaching the exchange. Using this order type, traders can synchronize order executions across different exchanges, such that high-frequency traders, even if they operate at the speed of light, can no-longer engage in latency arbitrage.}

\hspace{.5cm}

\section{Introduction}\label{Introduction}

Investors, who minimize the cost at which they acquire a given number of shares, have an incentive to break-up large orders into several smaller orders, which are then placed on different exchanges. Moreover, in order to avoid that a high frequency trader (HFT) can front-run some of their orders, investors have an incentive to place orders such that they are executed simultaneously across different exchanges. Such simultaneous executions are, however, difficult to implement in a world with random latencies.\footnote{That is, suppose an investor sends two buy orders to two different exchanges, and one of these orders, e.g. Order 1, reaches Exchange 1 some time before Order 2 reaches Exchange 2. This scenario allows a HFT, who detects the early execution of Order 1 on Exchange 1, to quickly buy on Exchange 2. In turn, the HFT can sell at a profit when the investor's Order 2 reaches Exchange 2. High frequency traders operate dedicated glass-fiber networks for this purpose. Such networks allow for (one way) latencies of roughly 4 milliseconds (ms) between Chicago and New York. At the same time, an investor, e.g. from Albany, faces a distribution of latencies: Albany-New York $(\mu=51ms,\sigma=28ms)$ and Albany-Chicago $(\mu=103ms,\sigma=25.7ms)$. That is, all orders sent from Albany to New York and to Chicago, which do not arrive within 4 ms of one-another, are subject to latency arbitrage. Using data from the NYSE and the CME, \citet{Bud15} show that such arbitrage opportunities are roughly worth 75 million USD annually for the trade in the SPY ETF alone.}   

In the present paper, we study latency arbitrage in a model where investors/traders buy and sell one homogenous asset on two geographically distinct exchanges. Trading is complicated by randomly varying latencies, and by the presence of high-frequency traders (HFTs), who enjoy lower latencies than all other market participants. Based on this model, we propose a new order type, which allows investors to specify the time at which their orders are executed after reaching the exchange. That is, we propose a market design, where traders can, in addition to choosing when to send orders, choose the time at which an order is executed after it reaches the exchange. In turn, we show that traders can use this order type to better synchronize order executions across exchanges, such that HFTs can no-longer engage in latency arbitrage.

\textsl{Related Literature:} \citet{Bud15} have recently argued that the competition for lower latencies among HFTs is ``a symptom of a flawed market design." In turn, to reduce the rents collected by HFTs, \citet{Bud15}, p. 1549, propose that exchanges should switch from ``continuous-time trading" to ``discrete-time trading." That is, \citet{Bud15} argue that exchanges should only execute orders at prescribed, discrete, points in time. Put differently, \citet{Bud15} propose to restrict traders' choices regarding the execution times of their trades. In the present paper we show that giving traders additional, rather than fewer, choice variables may also resolve the problem of latency arbitrage. Taking this view, ``discrete-time trading" may be an unnecessary constraint on the market place.


\citet{Bud15}, p. 1548, take the observation that correlations between the prices of homogenous assets, which trade on two different exchanges, break down at ``high-frequency time horizons" \emph{as an exogenous, empirical, fact}.\footnote{Put differently, \citet{Bud15} observe that the price of the SPY in Chicago is not perfectly correlated with the price of the SPY on the NYSE. Put yet differently, \citet{Bud15} show that the law of one price does not hold at very short time horizons. This phenomenon is (naturally) even more pronounced, e.g. \citet{Epp79}, in older data sets.} In turn, \citet{Bud15}, p. 1552, argue that ``[t]his correlation breakdown in turn leads to obvious mechanical arbitrage opportunities, available to whoever is fastest. For instance, at 1:51:39.590 PM, after the price of the ES [Chicago] has just jumped roughly 2.5 index points, the arbitrage opportunity is to buy SPY [NYSE] and sell ES [Chicago]." In the present paper, on the contrary, we start with a model where the break down in correlations results \textsl{endogenously} from the placement of large orders under random latency. That is, the observation ``[a]fter the price of the ES [Chicago] has just jumped roughly 2.5 index points, the arbitrage opportunity is to buy SPY [NYSE] and sell ES [Chicago]," has the following interpretation in our model: a trader, large enough to move the market by 2.5 index points, must have placed an order in a manner which creates an arbitrage opportunity. In turn, we show that creating such arbitrage opportunities is (i) costly for the trader and (ii) can be avoided if the trader uses the order type proposed here. That is, the arbitrage opportunities, upon which \citet{Bud15} build their argument for slowing markets, are no-longer present.

\citet{Bud15}, and \citet{Bud20} review\footnote{See also \citet{Sti14} and \citet{Lin17} for recent reviews on high-frequency trading, and \citet{Rot94,Rot02} for a broader market design perspective on the optimal time that markets open and close.} several alternative proposals aimed at reducing latency arbitrage. Many of these proposals suggest to Tobin-tax financial transactions, or to tax high frequency trading, or to tax low latency infrastructure. Other proposals argue for reductions in the speed at which orders are executed, or for reductions in the frequency with which markets open, or limits to the speed with which market participants can place/cancel orders. Yet different proposals suggest to introduce additional noise in the placement times of orders, to dilute the speed advantage of HFTs. Another branch of models suggest that fast traders should compete in a ``fast market," and that slow traders should trade in a ``slow market." These proposals have in common that they place additional restrictions on markets and market participants. The present paper thus offers an alternative perspective: we argue that latency arbitrage can be addressed by removing, rather than adding, to the restrictions that market participants face. 

\textsl{Organization:} Section \ref{Deterministic} studies a deterministic benchmark. Section \ref{Lat} introduces the latency friction, and shows that contemporary market designs, where orders are executed as soon as they reach the exchange, are prone to latency arbitrage; even if traders strategically delay the sending of orders. We also note that strategic order delay generates price distributions, which are in line with the empirical observations in \citet{Bud15}. Section \ref{Imp} proposes an order type, which helps traders to synchronize order executions across exchanges. Using recent latency data, Section \ref{Dis} illustrates the practical effectiveness of the order type that we propose. Section \ref{Con} concludes.

\section{Deterministic Benchmark}\label{Deterministic}

One asset is traded on two exchanges $m=L,S$. Each exchange has its own limit order book/excess demand function for the asset. The number/density of shares $f(P)$, which are on offer at each price $P$, differs across exchanges. To distinguish between the two exchanges, we assume that there is a large exchange $L$, which is more liquid than the smaller exchange $S$ in the sense that $f_L(P)>f_S(P)\forall P$.\footnote{The \citet{CME16}, p.3, estimates that its market for the SPY future is 7 times more liquid than the NYSE's market for the SPY ETF. The \citet{CME16}, p.3, also estimates that buying 100 Million worth of the S\&P 500 costs 1.25 basis points (BP) on the CME while the cost is 2 BP if the same amount of the SPY ETF is bought on the NYSE.} We also assume that, unless a large order is placed in a manner that brings prices into temporary disequilibrium, the market satisfies the law of one price $P_L=P_S=P_0$.\footnote{Appendix \ref{A1}, presents such a market place, consisting of two local markets/exchanges, each of which with a distinct (excess) demand function/limit-order book.}

A trader, who buys all shares offered for prices less or equal $P^*_m$ on exchange $m$, receives a quantity $X_m$:
\begin{eqnarray} X_m=\int_{P_0}^{P_m^*}f_m(P)dP, \quad m=L,S.  \end{eqnarray}
The cost of buying a bundle of stocks $X_m$ in market $m$ is thus:
\begin{eqnarray} E_m=\int_{P_0}^{P_m^*}Pf_m(P)dP, \quad m=L,S. \label{e00}\end{eqnarray}
To minimize the cost, of acquiring a given bundle of stocks $X$, the trader buys shares on both exchanges:
\begin{eqnarray} \min_{P^*_L,P^*_S}\Big\{\int_{P_0}^{P_L^*}Pf_L(P)dP+\int_{P_0}^{P_S^*}Pf_S(P)dP\Big\} \quad s.t. \quad X_L+X_S=X. \label{cost}\end{eqnarray} 
The first-order conditions to problem (\ref{cost}) imply that:
\begin{eqnarray} P^*_L=P^*_S=P^*.\label{FOC} \end{eqnarray}
Hence, we have:
\begin{lemma} Large traders split orders between both marketplaces such that the law of one price is not violated.\end{lemma}

\section{Random Latency}\label{Lat}
Suppose now that the trader communicates his orders via a telecommunication network with random latency. That is, orders to exchanges L and S may be delayed such that one order is executed earlier than the other. A high-frequency trader (HFT) can exploit this. That is, once he observes that, e.g., the price on the small exchange increases to $P^*_S$, he knows from equation (\ref{FOC}) that there is an order $X^*_L,P^*_L$ on its way to exchange $L$. He thus quickly buys the quantity $X^*_{L}$ on the large exchange, in order to sell at a higher price once the trader's delayed order arrives at exchange $L$.\footnote{If the HFT's order arrives after the order of the trader, the HFT gets no fill and cancels the order. That is, the HFT acts as a pure arbitrageur in our model.} This yields a rent for the HFT:
\begin{eqnarray} P_L^*X_L-\int_{P_0}^{P_L^*}Pf_L(P)dP>0, \end{eqnarray}
and adds to the cost at which an investor acquires stocks on exchange $L$.\footnote{Note that, even if the trader knew that his order is front run buy the HFT with probability one, he would still buy from the HFT: executing the whole order on just one exchange would increase the (short-run) price on that exchange beyond the price $P^*$ that he is paying when he buys from the HFT. Note also that a HFT, who acts as a pure arbitrageur, carries no inventory. That is, he will not buy more than what the investor is "willing" to buy from him.}

\begin{tikzpicture}
	\begin{linesaxes}
		\node (I) at (axis cs:20, 60) {I}; 
		\node (L) at (axis cs:70, 80) {L}; 
		\node (S) at (axis cs:60, 20) {S}; 
		
		\node[anchor=west] (HFT) at (axis cs:67,50) {HFT}; 
		
		\node[inner sep=0pt] (TS) at (axis cs:40, 40) {}; 
		
		\draw[->, very thick] (I) -- (L); 
		\draw[->, very thick] (L) -- (S); 
		\draw[->, very thick] (I) -- (TS); 
		\draw[dotted, very thick] (TS) -- (S); 
	\end{linesaxes}
\end{tikzpicture}
\begin{tikzpicture}
	\begin{linesaxes}
		\node (I) at (axis cs:20, 60) {I}; 
		\node (L) at (axis cs:70, 80) {L}; 
		\node (S) at (axis cs:60, 20) {S}; 
		
		\node[anchor=west] (HFT) at (axis cs:67,50) {HFT}; 
		
		\node[inner sep=0pt] (TL) at (axis cs:55, 74) {}; 
		
		\draw[->, very thick] (S) -- (L); 
		\draw[->, very thick] (I) -- (S); 
		\draw[->, very thick] (I) -- (TL); 
		\draw[dotted, very thick] (TL) -- (L); 
	\end{linesaxes}
\end{tikzpicture}

\begin{tikzpicture}
	\begin{linesaxes}
		\node (I) at (axis cs:20, 60) {I}; 
		\node (L) at (axis cs:70, 80) {L}; 
		\node (S) at (axis cs:60, 20) {S}; 
		
		\node[anchor=west] (HFT) at (axis cs:69,68) {HFT}; 
		
		\node[inner sep=0pt] (LS) at (axis cs:66.5, 60) {}; 
		
		\draw[->, very thick] (I) -- (L); 
		\draw[->, very thick] (I) -- (S); 
		\draw[->, very thick] (L) -- (LS); 
		\draw[dotted, very thick] (LS) -- (S); 
	\end{linesaxes}
\end{tikzpicture}
\begin{center} Diagram 1: Latency Arbitrage Triangle\end{center}

Diagram 1 summarizes the three possible outcomes faced by an investor $I$. Top left: the large exchange $L$ reveals the trade to the HFT. For this case, we denote the investor's total expenditures by $E_L$. Top right, the small exchange $S$ reveals the trade. For this case, we denote the investor's total expenditures by $E_S$. Bottom left, orders are executed simultaneously. In this case, the investor pays $E_{sim}$. 

Using our assumptions on market liquidity, namely $f_S(P)<F_L(P)$, we can rank these outcomes:
\begin{lemma} $E_{sim}<E_L<E_S$. \label{Expenditurelemma}\end{lemma}
\begin{proof}
\begin{eqnarray} E_{sim}=\int_{P_0}^{P^*}Pf_L(P)dP+\int_{P_0}^{P^*}Pf_S(P)dP<\int_{P_0}^{P^*}Pf_L(P)dP+P^*X_S=E_L \end{eqnarray}
\begin{eqnarray} E_L=\int_{P_0}^{P^*}Pf_L(P)dP+P^*X_S<P^*X_L+\int_{P_0}^{P^*}Pf_S(P)dP=E_S \end{eqnarray}
\end{proof}

\subsection{Optimal Order Delay}\label{Delay}

In this section we study how investors can cope with random latencies under the contemporary market design, where orders are executed as soon as they arrive at the exchange. More precisely, we show that investors have an incentive to strategically delay orders such that trades are revealed more often on the liquid exchange. Put differently, the contemporary market design forces traders to tradeoff early execution on one exchange against early execution on the other exchange. Simultaneous execution of orders, however, is very difficult to implement. 
	
The observation that traders delay orders to the small exchange, to increase the probability with which trades are revealed early on the large exchange, is in line with the empirical evidence in \citet{Bud15}, p. 1569, who find that ``[t]he majority (88.56 percent) of arbitrage opportunities in our data set are initiated by a price change in ES [Chicago], with the remaining 11.44 percent initiated by a price change in SPY [NYSE]."\footnote{\citet{CME16}, p. 3, discuss that the Chicago market for the S\&P 500 is significantly more liquid than that of the NYSE.} Moreover, they remark that this ``is consistent with the practitioner perception that the ES [Chicago] market is the center for price discovery in the S\&P 500 index." Unlike \citet{Bud15}, however, we do not treat the early executions in Chicago as an exogenous empirical fact. Early executions in Chicago are an endogenous feature of our model. 

In turn, with these observations in place, we present a simple solution to the problem of latency arbitrage in Section \ref{Imp}.

We denote order delays to the small exchange, relative to the order to the large exchange, by $\delta\in\mathcal{R}$.\footnote{That is, $\delta=10$ means that the order to exchange $S$ is delayed by 10 milliseconds (ms), and $\delta=-20$ means that the order to the large exchange is sent 20 ms after the order to exchange $S$ was sent.} To minimize the expected cost of purchasing a given quantity of the asset he chooses $\delta$ such that:
\begin{eqnarray} \min_{\delta}\Big\{\pi_{sim}(\delta,H)E_{sim}+\pi_L(\delta,H)E_L+\pi_S(\delta,H)E_S\Big\}, \quad \pi_{sim}+\pi_L+\pi_S=1. \label{Delay1}\end{eqnarray}Where $\pi_{sim},\pi_L,\pi_S$ are the probabilities with which orders are executed simultaneously, or revealed on the large exchange or on the small exchange, respectively. To evaluate (\ref{Delay1}), we first work with general probabilities. In turn, we assume that latencies are normally distributed, and solve for traders' optimal order delay explicitly.

Regarding probabilities we assume that
$\pi_{L,\delta}:=\frac{\partial\pi_{L}}{\partial \delta}\geq 0$,
$\pi_{S,\delta}:=\frac{\partial\pi_{S}}{\partial \delta}\leq 0$.
Since $\pi_{sim}+\pi_L+\pi_S=1$, we have
$\pi_{sim,\delta}:=\frac{\partial\pi_{sim}}{\partial
\delta}\gtreqqless 0$. The first order condition
to problem (\ref{Delay1}) is:
\begin{eqnarray} \pi_{sim,\delta}(E_{sim}-E_L)=\pi_{S,\delta}(E_{L}-E_{S}), \label{FOCa} \end{eqnarray}
respectively
\begin{eqnarray} \pi_{L,\delta}(E_L-E_{sim})+\pi_{S,\delta}(E_{S}-E_{sim})=0. \label{FOCb} \end{eqnarray}
We thus have\footnote{If latency has compact support, there exists an interior optimal delay $\delta^*$ as long as $H>0$.}
\begin{lemma} At an interior solution $\delta^*$, we have $\pi_{sim,\delta}<0$ and $|\pi_{L,\delta}|>|\pi_{S,\delta}|$.\label{lemgen}\end{lemma}
\begin{proof} $\pi_{sim,\delta}<0$ follows from (\ref{FOCa}) and Lemma \ref{Expenditurelemma}. $|\pi_{L,\delta}|>|\pi_{S,\delta}|$ follows from (\ref{FOCb}) and Lemma \ref{Expenditurelemma}.  \end{proof}
Lemma \ref{lemgen} indicates that traders do not maximize the probability of simultaneous execution. Instead, they lean towards execution on the large exchange as illustrated in Diagram 2. In Appendix \ref{normaldelay} we assume that latencies are normally distributed, and solve explicitly for $\delta^*$.

\begin{tikzpicture}
	\begin{distributionsaxes}
		\addplot [very thick, red] {gauss(x, 2.25, 0.7)}; 
		\addplot [very thick, blue] {gauss(x, 3.25, 0.7)}; 
		
		\draw [very thick, dotted, red] (axis cs:2.25,0) -- (axis cs:2.25, \valueA); 
		\draw [very thick, dotted, blue] (axis cs:3.25,0) -- (axis cs:3.25, \valueA); 
		
		\node[below, red] at (axis cs:3.7, 0.2) {$\Phi_S$}; 
		\node[below, blue] at (axis cs:4.7, 0.2) {$\Phi_L$}; 
		
		\node[below, red] at (axis cs:2.25, 0) {$\mu_S$}; 
		\node[below, blue] at (axis cs:3.25, 0) {$\mu_L$}; 
	\end{distributionsaxes}
\end{tikzpicture}
\begin{tikzpicture}
	\begin{distributionsaxes}
		\addplot [very thick, red] {gauss(x, 4, 0.7)}; 
		\addplot [very thick, blue] {gauss(x, 2.25, 0.7)}; 
		
		\draw [very thick, dotted, blue] (axis cs:2.25,0) -- (axis cs:2.25, \valueA); 
		\draw [very thick, dotted, red] (axis cs:4,0) -- (axis cs:4, \valueA); 
		
		\node[below, red] at (axis cs:5.45, 0.2) {$\Phi_S$}; 
		\node[below, blue] at (axis cs:3.7, 0.2) {$\Phi_L$}; 
		
		\node[below, red] at (axis cs:4, 0) {$\mu_S$+$\delta$}; 
		\node[below, blue] at (axis cs:2.25, 0) {$\mu_L$}; 
	\end{distributionsaxes}
\end{tikzpicture}
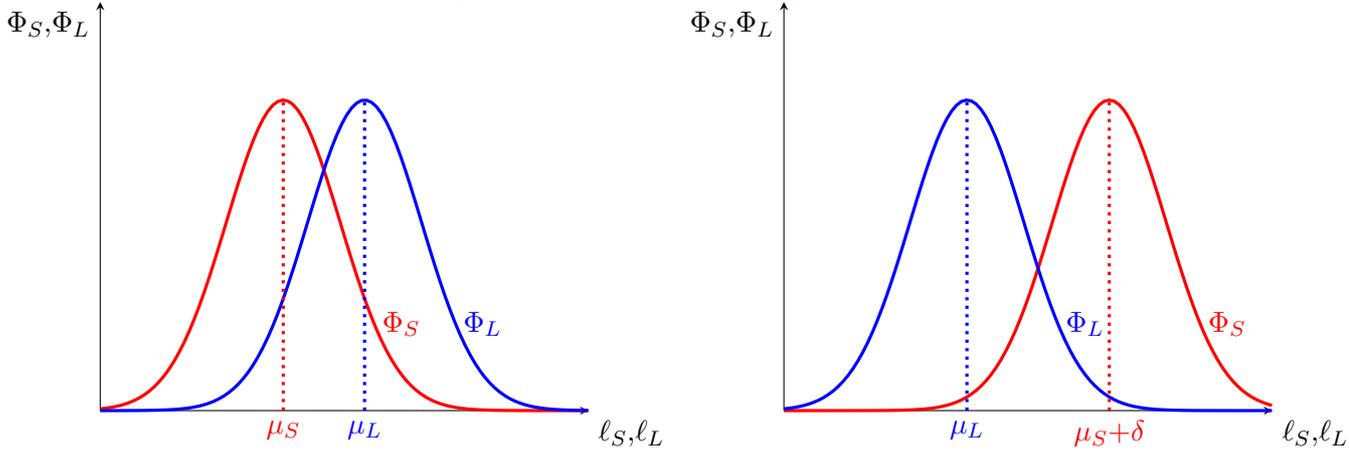
\begin{center} Diagram 2: Delaying Orders. Traders have an incentive to delay orders to the small exchange. The optimal delay $\delta^*$ does not maximize the probability of simultaneous execution.\end{center}

\section{Synchronized Order Placement}\label{Imp}

Suppose now that orders can be accompanied by a time identifier, which specifies the exact time $T$ at which the order is added to the exchange's order-book/executed. That is, orders are sent out to the exchanges at time $t=0$, with an identifier $T_m>0, m=L,S$, indicating the time at which the exchange adds the respective order to its order book, i.e. the time when the trade is executed. That is, a trade is implemented via the following time line:
\begin{enumerate}
    \item Orders are sent to exchanges. In addition to price and quantity, these orders also specify the exact time of execution/addition to the limit order book.
    \item Once orders arrive at the respective exchanges, they are not executed/placed until the specified placement time is reached. Exchanges are not allowed to publish the receipt of these orders until the placement time has been reached. 
    \item If an order arrives at the exchange after the desired placement time, it is placed immediately.
\end{enumerate}

\begin{lemma} Under the new order type, traders can choose placement times such that simultaneous execution is ensured. \end{lemma}
\begin{proof}
We note that:
\begin{eqnarray} \pi_{sim}&\geq&P(l_S\leq T+H)P(l_L\leq T+H)+P(|l_S-l_L|\leq H)(1-P(l_S\leq T+H)P(l_L\leq T+H))\nonumber\\
&\geq& P(l_S\leq T+H)P(l_L\leq T+H) \label{sim}\end{eqnarray}
\begin{eqnarray}\pi_{sim}=1-\pi_L-\pi_S \label{simb}\end{eqnarray}
Equations (\ref{sim}) with (\ref{simb}), and the fact that $\pi_L,\pi_S\geq 0$ and $\lim_{T\rightarrow\infty}P(l_S\leq T+H)P(l_L\leq T+H)=1$, imply: $\lim_{T\rightarrow\infty}\pi_{sim}=1,$ $\lim_{T\rightarrow\infty}\pi_L=0$, and $\lim_{T\rightarrow\infty}\pi_S=0$.\end{proof}

That is, simultaneous execution can be ensured via the extended order type: instead of delaying messages, to tradeoff the probability of early execution on one exchange with the probability of early execution on the other. Put differently, increases in placement time $T$ simultaneously reduce the probabilities $\pi_L$ and $\pi_S$ with which latency arbitrage occurs, and unambiguously increase the probability of simultaneous order placement.\footnote{Using latency data by \citet{won21}, 27.05.2021 at 9:22 GMT, we note that the present order type can ensure simultaneous order placement, even for very bad connections. For example, the \textsl{mean} latencies from Kampala to Manhattan and Chicago are both roughly 440 ms. The \textsl{highest} observed latency from Kampala to Manhattan (latencies to, e.g. servers in New Jersey, are similar) is 640 ms. The \textsl{highest} latency from Kampala to Chicago was 671 ms. Hence, a trader from Kampala can ensure simultaneous execution of his orders by setting $T=671$ ms, i.e. 0.67 seconds. This means that the mean delay necessary to ensure simultaneous execution is roughly 0.2 seconds. Within the US, we find that Knoxville, Tennessee, puts up the highest latency numbers (the maximum observed latencies to Manhattan and Chicago are 70 ms and 80 ms respectively). Accordingly, Knoxville based traders can ensure simultaneous execution of orders by setting T=0.08 seconds. All other US based traders, who enjoy a better connection, can ensure simultaneous execution by choosing even lower values for $T$. Likewise, traders from London and Frankfurt, could choose $T=80$ ms and $T=95$ ms respectively, to ensure simultaneous order execution.}

Diagram 3 illustrates how choosing execution time $T$ increases the probability of simultaneous order execution. The shaded area on the left indicates the probability of simultaneous order placement when traders can specify execution time $T$, conditional on a latency realization $\tilde{l_S}.$ The shaded area on the right indicates the probability of simultaneous order placement when traders cannot specify execution time $T$, conditional on a latency realization $\tilde{l_S}.$

\begin{tikzpicture}
	\begin{distributionsaxes}
		\addplot [fill=blue!20, draw=none, domain=0:4.75] {gauss(x, 3.25, 0.7)} \closedcycle; 
		
		\addplot [very thick, red] {gauss(x, 2.25, 0.7)}; 
		\addplot [very thick, blue] {gauss(x, 3.25, 0.7)}; 
		
		\node[below, red] at (axis cs:3.5, 0.3) {$\Phi_S$}; 
		\node[below, blue] at (axis cs:4.5, 0.3) {$\Phi_L$}; 
		
		\node[below, red] at (axis cs:2.25, 0) {$\tilde{\ell}_S$}; 
		\draw [very thick, red, yshift=-3pt] (axis cs:2.25, 0) -- (axis cs:2.25, 0.03); 
		\node[below, blue] at (axis cs:4.75, 0) {$T$+$H$}; 
		
		\draw [very thick, blue] (axis cs:4.75,0) -- (axis cs:4.75, \valueA); 
	\end{distributionsaxes}
\end{tikzpicture}
\begin{tikzpicture}
	\begin{distributionsaxes}
		\addplot [fill=blue!20, draw=none, domain=1.75:2.5] {gauss(x, 3.25, 0.7)} \closedcycle; 
		
		\addplot [very thick, red] {gauss(x, 2.25, 0.7)}; 
		\addplot [very thick, blue] {gauss(x, 3.25, 0.7)}; 
		
		\node[below, red] at (axis cs:3.5, 0.3) {$\Phi_S$}; 
		\node[below, blue] at (axis cs:4.5, 0.3) {$\Phi_L$}; 
		
		\node[below, red] at (axis cs:2.125, 0) {$\tilde{\ell}_S$}; 
		\draw [very thick, red, yshift=-3pt] (axis cs:2.125, 0) -- (axis cs:2.125, 0.03); 
		
		\draw [yshift=0.25cm, latex-latex, red](axis cs:1.75,0) -- node [fill=blue!20] {\footnotesize $2H$} (axis cs:2.5,0); 
	\end{distributionsaxes}
\end{tikzpicture}
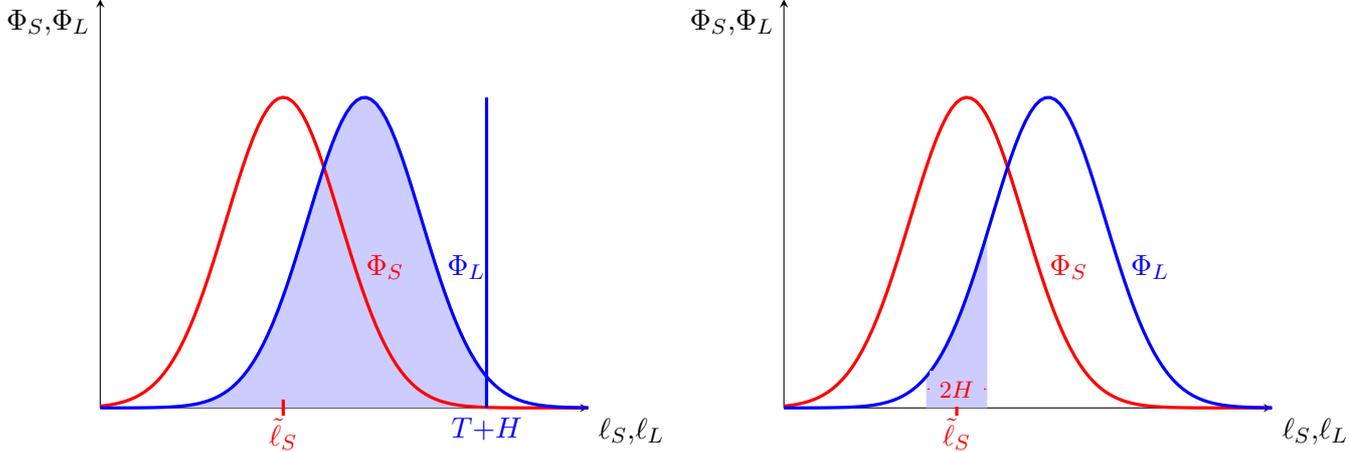
\begin{center} Diagram 3: Probability of Synchronized Order Placement: the shaded areas illustrate the probability of simultaneous order placement, conditional on a latency $\tilde{l}_S$. The probability of synchronized order placement (on the left), where traders can choose the execution time $T$ is much larger than under the current market design (on the right), where orders are executed immediately after reaching the exchange.\end{center}

\subsection{Calibration}\label{Dis}

To illustrate the advantage of the order type proposed here, let us consider the case of an Albany (New York State) based investor, who trades the S\&P 500 in New York City and Chicago:

\begin{enumerate}
    \item The HFT's one-way Chicago to New York time is roughly H=4 ms

    \item Latencies\footnote{We use latency data provided by \citet{won21}, 27.05.2021, 9:22 GMT.} are distributed: Albany-NY $(\mu_S=51,\sigma_S=28)$ and Albany-Chicago $(\mu_L=103,\sigma_L=25.7)$. 

    \item Relative excess cost $\frac{(E_L-E_{sim})}{(E_{S}-E_{sim})}$: \citet{CME16} estimate the price impact of placing 100 million SPY order as 1,25 BP. Placing the same order on the NYSE has an estimated impact of 2BP. Using the the linear model in Appendix \ref{A1}, we thus have $\frac{E_S-E_{sim}}{E_L-E{sim}}=\frac{2}{1.25}=1.6$.
\end{enumerate}

In Appendix \ref{Cal}, we use these data to compute execution probabilities for (i) the myopic scenario, where orders are simply sent simultaneously to the exchanges (ii) the model of Section \ref{Delay}, where traders strategically delay the sending of orders and (iii) for the model of Section \ref{Imp}, where traders can specify the execution times of their orders. First, we find that only 96\% of trades are subject to latency arbitrage, when orders are simply sent simultaneously by investors. Second, for the model of optimal order delay of Section \ref{Delay} and Appendix \ref{normaldelay}, we find that a large majority, i.e. roughly 98\% of trades, are first revealed in Chicago. This finding is in line with the empirical evidence provided in \citet{Bud15}.\footnote{\citet{Bud15}, p. 1569, who find that ``[t]he majority (88.56 percent) of arbitrage opportunities in our data set are initiated by a price change in ES [Chicago], with the remaining 11.44 percent initiated by a price change in SPY [NYSE]."} Moreover, the incentive to skew early executions to the large exchange is so strong that the probability of simultaneous execution is only 1\%, i.e. even lower than the 4\% of simultaneous executions that we observed in the scenario where orders were not strategically delayed. Finally, once traders can use the order type proposed in Section \ref{Imp}, they can, e.g. set T=150 ms, to ensure that over 99\% of trades are executed simultaneously. Put differently, given that the mean latency from from Albany to Chicago is 103 ms, a mean delay of 0.047 seconds ensures that latency arbitrage is no-longer possible.

\paragraph{Technical Feasibility:} the order type proposed here, requires that both exchanges use reasonably precise clocks. Regarding time measurement, we not that recent MIFID II regulation, \citet{Euc16}, requires that all financial markets transactions within the EU, which are related to high-frequency trading, are recorded with a precision of at least 100 microseconds, i.e. $0.1$ ms. That is, there would be a maximum time span of $\pm 0.2ms$, with which orders would be placed on two separate exchanges once the order type proposed here is used. This error is an order of magnitude smaller than the time frames that any HFT\footnote{Sending a message from Chicago to New York 3.66 ms, if it travels at the speed of light. That is, possible time measurement errors $\xi$ within $\pm .1$ ms, which would change execution times from exactly $T$ to $T\pm \xi$, would not change our analysis. Indeed time frames of $0.1$ ms are still small in the context of European markets, where geographic distances are smaller. That is, it takes light roughly 1 ms to travel from London to Frankfurt.} could exploit.


\section{Conclusion}\label{Con}
We propose an order type, which allows traders to specify the time at which their orders are executed after reaching the exchange. Using this order type, traders can synchronize order executions across exchanges such that HFTs can no-longer engage in latency arbitrage. Put differently, the order type proposed here allows large traders to place orders such that the law of one price holds, even at ``high-frequency time horizons."

Earlier proposals in the literature, aimed at reducing latency arbitrage, require taxes on financial transactions, or place restrictions on the speed at which trades are executed, orders placed, or prices quoted. The present paper thus offers an alternative perspective, where a relaxation, rather than a tightening, of constraints helps traders, and thus the entire market, to avoid the cost of latency arbitrage. 
\newpage
\appendix

\section{Linear Model}\label{A1}

We study a linear demand system

\begin{eqnarray} P_S=a-bX_S, \quad a,b>0 \label{S} \\
P_L=c-dX_L, \quad c,d>0\label{L}\\
X_S+X_L=\bar{X}\\
P_S=P_L=P_0,  \end{eqnarray}
which rewrites as:
\begin{eqnarray} X_S=\frac{a+d\bar{X}-c}{b+d} \quad X_L=\frac{-a+b\bar{X}+c}{b+d} \nonumber\\
P_S=a-b\frac{a+d\bar{X}-c}{b+d}\nonumber\\
P_L=c-d\frac{-a+b\bar{X}+c}{b+d}. \nonumber\end{eqnarray}

If a trader buys a number of shares $\tilde{X}$, we obtain new (long-run) prices and quantities
\begin{eqnarray} X_S=\frac{a+d(\bar{X}-\tilde{X})-c}{b+d}, \quad X_L=\frac{-a+b(\bar{X}-\tilde{X})+c}{b+d}\nonumber\\
P_S=a-b\frac{a+d(\bar{X}-\tilde{X})-c}{b+d}\nonumber\\
P_L=c-d\frac{-a+b(\bar{X}-\tilde{X})+c}{b+d}. \nonumber\end{eqnarray}

\paragraph{Strategy of HFT and the Limit Order Book}

For the current demand/supply system we have:
\begin{eqnarray} \frac{dX_S}{dP_S}=-\frac{1}{b}, \quad \frac{dX_L}{dP_L}=-\frac{1}{d}\nonumber\end{eqnarray}
the cost of purchasing a quantity of shares $\tilde{X}$ is thus:
\begin{eqnarray} \tilde{X}=-\int_{P_0}^{P^*}\frac{dX_S}{dP_S}dP_S-\int_{P_0}^{P^*}\frac{dX_L}{dP_L}dP_L=(P^*-P_0)(\frac{1}{b}+\frac{1}{d}) \nonumber \end{eqnarray}
Expenditures, in the case of simultaneous execution, are:
\begin{eqnarray} \tilde{E}_{sim}=-\int_{P_0}^{P^*}P_S\frac{dX_S}{dP_S}dP_S-\int_{P_0}^{P^*}P_L\frac{dX_L}{dP_L}dP_L=\frac{1}{2}(P^{*2}-P^2_0)(\frac{1}{b}+\frac{1}{d})\nonumber  \end{eqnarray}
Expenditures, in the case where the trade is revealed on exchange $L$, are:
\begin{eqnarray} \tilde{E}_L=-P^*\int_{P_0}^{P^*}\frac{dX_S}{dP_S}dP_S-\int_{P_0}^{P^*}P_L\frac{dX_L}{dP_L}dP_L=P^{*2}\frac{1}{b}-P^*P_0\frac{1}{b}+\frac{1}{2}(P^{*2}-P_0^2)\frac{1}{d}\nonumber\end{eqnarray}
Expenditures, in the case where the trade is revealed on exchange $L$, are:
\begin{eqnarray} \tilde{E}_S=-\int_{P_0}^{P^*}\frac{dX_S}{dP_S}dP_S-P^*\int_{P_0}^{P^*}\frac{dX_L}{dP_L}dP_L=
P^{*2}\frac{1}{d}-P^*P_0\frac{1}{d}+\frac{1}{2}(P^{*2}-P_0^2)\frac{1}{b}\nonumber\end{eqnarray}
Hence we have
\begin{eqnarray} E_L-E_{sim}=\frac{1}{2b}(P^*-P_0)^2>0\nonumber\\
E_S-E_{sim}=\frac{1}{2d}(P^*-P_0)^2>0\nonumber\\
\frac{E_S-E_{sim}}{E_L-E_{sim}}=\frac{b}{d}\label{Ratio}\nonumber
\end{eqnarray}
Taking into account the \citet{CME16}, p.3, estimate that a
purchase worth 100 Million in the SPY increases prices by $\Delta
P_L=1.25$ BP and $\Delta P_S=2$ BP respectively. Moreover, for
demands (\ref{S}) and (\ref{L}) we have $d=\frac{\Delta
P_L}{\Delta X_L}$ as well as $b=\frac{\Delta P_S}{\Delta X_S}$.
Hence, we can recall (\ref{Ratio}) to obtain
$\frac{E_S-E_{sim}}{E_L-E_{sim}}=\frac{b}{d}\approx
\frac{2}{1.25}=1.6$.

\section{Normally Distributed Latency}\label{normaldelay}
Delivery times for normally distributed network noise are:
\begin{eqnarray}
l_S=\mu_S+\delta+\sigma_S\xi \quad \xi\sim\mathcal{N}(0,1)\\
l_L=\mu_L+\sigma_L\varepsilon \quad \varepsilon\sim\mathcal{N}(0,1)
\nonumber\end{eqnarray}
The HFT's order delivery time is $H>0$. To work with these
latencies, we define $x:=l_S-l_L$ and
$\gamma:=E[l_S-l_L]=\mu_S-\mu_L+\delta$. Moreover, we assume that
latencies are uncorrelated such that
$\alpha:=\frac{1}{Var(l_S-l_L)}=\frac{1}{\sigma_S^2+\sigma_L^2}$.
We denote the cumulative standard normal distribution function by
$\Phi()$ and it's derivative by $\phi()$. The probabilities of
early revelation on exchanges $L,S$ as well as the probability of
simultaneous execution are thus:
\begin{eqnarray}
\pi_{L}=P(l_S-l_L>H)=1-\Phi(\sqrt{\alpha}(-\gamma+H))\label{Normal1}\\
\pi_{S}=P(l_S-l_L<-H)=\Phi(\sqrt{\alpha}(-\gamma-H))\label{Normal2}\\
\pi_{sim}=P(|l_S-l_L|<H)=1-\pi_L-\pi_S\label{Normal3}
\end{eqnarray}
Given (\ref{Normal1}) and (\ref{Normal2}), the first order condition for optimal order delay (\ref{FOCb}) can be rewritten as:
\begin{eqnarray} \phi(\sqrt{\alpha}(-\gamma+H))(E_L-E_{sim})=\phi(\sqrt{\alpha}(-\gamma-H))(E_{S}-E_{sim}). \label{FOCbnormal} \end{eqnarray}
Recalling $\phi(z)=\frac{1}{\sqrt{2\pi}}e^{-\frac{z^2}{2}}$, we solve (\ref{FOCbnormal}):
\begin{lemma} Orders to the small exchange are delayed such that $\delta^*=\mu_L-\mu_S+\frac{\sigma_S^2+\sigma_L^2}{2H}ln(\frac{E_S-E_{sim}}{E_L-E_{sim}}),$ and $\gamma^*(\delta^*)=\frac{\sigma_S^2+\sigma_L^2}{2H}ln(\frac{E_S-E_{sim}}{E_L-E_{sim}})>0$, and $\pi_L(\delta^*)>\pi_S(\delta^*)$. \label{L1}\end{lemma}

\section{Early Execution Probabilities}\label{Cal}

\emph{Orders without strategic delay:} without delay, i.e. $\delta=0$ such that $\gamma=\mu_S-\mu_L+\delta=-52$, we have $\pi_S\approx\Phi(\frac{48}{38})\approx 0.89,\pi_L= 1-\Phi(-\frac{56}{38})\approx 0.07$ and $\pi_{sim}\approx 0.04$.

\emph{Orders with strategic delay:} recalling Appendix \ref{normaldelay}, we have optimal delay $\delta^*=\mu_L-\mu_S+\frac{\sigma_S^2+\sigma_L^2}{2H}ln(\frac{E_S-E_{sim}}{E_L-E_{sim}}),$ and $\gamma^*(\delta^*)=\frac{\sigma_S^2+\sigma_L^2}{2H}ln(\frac{E_S-E_{sim}}{E_L-E_{sim}})\approx \frac{1444}{8}ln(1.6)\approx 84.8$. Recalling that $\alpha=\frac{1}{Var(l_S-l_L)}=\frac{1}{\sigma_S^2+\sigma_L^2}$, we have $\sqrt{\alpha}=\sqrt{\frac{1}{28^2+25.7^2}}\approx \frac{1}{38}$. Substitution into (\ref{Normal1})-(\ref{Normal3}) yields $\pi_S\approx\Phi(\frac{-88.8}{38})\approx 0.01,\pi_L\approx1-\Phi(\frac{-80.8}{38})\approx 0.98$ and thus $\pi_{sim}\approx 0.01$.

\emph{Orders with strategic placement time:} recalling that $\pi_{sim}\geq P(l_S\leq T+H)P(l_L\leq T+H)+P(|l_S-l_L|\leq H)(1-P(l_S\leq T+H)P(l_L\leq T+H))=\Phi(\frac{103}{5.2})\Phi(\frac{51}{5})+0.04(1-0.99)\approx 0.99$. Hence choosing, e.g. $T=150$ ms, yields $\pi_{sim}\approx 0.99$, i.e. effectively ensures simultaneous execution of orders.

\newpage

\addcontentsline{toc}{section}{References}
\markboth{References}{References}
\bibliographystyle{apalike}
\bibliography{References}

\begin{thebibliography}{}

\bibitem[Aquilina et~al., 2020]{Bud20}
Aquilina, M., Budish, E., and O'Neill, P. (2020).
\newblock Quantifying the high-frequency trading ``arms race'': A simple new
  methodology and estimates.
\newblock {\em {Workingpaper}}, pages 1--94.

\bibitem[Budish et~al., 2015]{Bud15}
Budish, E., Crampton, P., and Shim, J. (2015).
\newblock The high-frequency trading arms race: Frequent batch auctions as a
  market design response.
\newblock {\em {Quarterly Journal of Economics}}, 130(4):1547--1621.

\bibitem[CME-Group, 2016]{CME16}
CME-Group (2016).
\newblock The big picture: A cost comparison of futures and etfs.
\newblock {\em
  {https://www.cmegroup.com/trading/equity-index/files/a-cost-comparison-of-futures-and-etfs.pdf}},
  2. Edition:1--16.

\bibitem[Epps, 1979]{Epp79}
Epps, T. (1979).
\newblock Comovements in stock prices in the very short run.
\newblock {\em {Journal of the American Statistical Association}},
  74(366):291--298.

\bibitem[European-Commission, 2016]{Euc16}
European-Commission (2016).
\newblock Supplementing directive 2014/65/eu of the european parliament and of
  the council with regard to regulatory technical standards for the level of
  accuracy of business clocks.
\newblock {\em
  {$https://eur-lex.europa.eu/legal-content/EN/TXT/?uri=CELEX\%3A32017R0574$}}.

\bibitem[Linton and Mahmoodzadeh, 2017]{Lin17}
Linton, O. and Mahmoodzadeh, S. (2017).
\newblock Implications of high-frequency trading for security markets.
\newblock {\em {Annual Review of Economics}}, pages 237--259.

\bibitem[Roth and Ockenfels, 2002]{Rot02}
Roth, A.~E. and Ockenfels, A. (2002).
\newblock Last-minute bidding and the rules for ending second-price auctions:
  Evidence from ebay and amazon auctions on the internet.
\newblock {\em {American Economic Review}}, 92:1093--1103.

\bibitem[Roth and Xing, 1994]{Rot94}
Roth, A.~E. and Xing, X. (1994).
\newblock Jumping the gun: Imperfections and institutions related to the timing
  of market transactions.
\newblock {\em {American Economic Review}}, 84:992--1044.

\bibitem[Stiglitz, 2014]{Sti14}
Stiglitz, J.~E. (2014).
\newblock Tapping the brakes: Are less active markets safer and better for the
  economy.
\newblock {\em {Workingpaper}}, pages 1--19.

\bibitem[wondernetwork, 2021]{won21}
wondernetwork (2021).
\newblock Global ping statistics.
\newblock {\em {https://wondernetwork.com/pings/}}.

\end{thebibliography}

\end{document}